\documentclass[a4paper]{article}

\usepackage{amssymb,amsmath}
\usepackage{stmaryrd}
\usepackage{cmll}
\usepackage{proof}
\usepackage[all]{xy}
\usepackage{times}
\usepackage{a4wide}
\usepackage[usenames]{xcolor}
\usepackage[active]{srcltx}
\usepackage{url}

\usepackage{xcolor}
\newif{\ifMarginalComments}
\MarginalCommentstrue
\newcounter{ncomm}

\def\ie{\textit{i.e.}}
\def\eg{\textit{e.g.}}

\newenvironment{varitemize}
{
\begin{list}{\labelitemii}
{\setlength{\itemsep}{0pt}
 \setlength{\topsep}{0pt}
 \setlength{\parsep}{0pt}
 \setlength{\partopsep}{0pt}
 \setlength{\leftmargin}{15pt}
 \setlength{\rightmargin}{0pt}
 \setlength{\itemindent}{0pt}
 \setlength{\labelsep}{5pt}
 \setlength{\labelwidth}{10pt}
}}
{
 \end{list} 
}

\newcommand{\refsect}[1]{Sect.~\ref{sec:#1}}

\newcommand{\reftab}[1]{Table~\ref{tab:#1}}
\newcommand{\refdef}[1]{Definition~\ref{def:#1}}
\newcommand{\refth}[1]{Theorem~\ref{th:#1}}
\newcommand{\reflemma}[1]{Lemma~\ref{lemma:#1}}

\newcommand{\st}{\mathrel{|}} 
\newcommand{\gpipe}{\mathrel{\big |}} 
\newcommand{\dom}{\mathrm{dom}}
\newcommand{\op}[1]{#1^{\mathrm{op}}} 
\newcommand{\bigo}{\mathcal O} 

\newcommand{\NN}{\mathbb{N}} 
\newcommand{\Nat}{\NN} 
\newcommand{\natone}{n} 
\newcommand{\Words}{\mathbb W} 
\newcommand{\strone}{\xi} 
\newcommand{\strtwo}{\upsilon} 
\newcommand{\emstr}{\varepsilon} 
\newcommand{\len}[1]{|#1|} 

\newcommand{\IChan}{\mathcal I} 
\newcommand{\ichan}{i} 
\newcommand{\OChan}{\mathcal O} 
\newcommand{\ochan}{o} 
\newcommand{\Vis}{\mathcal A_v} 
\newcommand{\Act}{\mathcal A} 
\newcommand{\acone}{\alpha} 
\newcommand{\sa}{\tau} 
\newcommand{\lsysone}{\mathcal S}
\newcommand{\lsystwo}{\mathcal T}
\newcommand{\Web}[1]{|#1|} 
\newcommand{\stateone}{s} 
\newcommand{\statetwo}{t} 
\newcommand{\initone}{s_0} 
\newcommand{\inittwo}{t_0} 
\newcommand{\trans}{\mathrm{trans}}
\newcommand{\ls}[3]{#1\stackrel{#2}{\longrightarrow}#3} 
\newcommand{\lstaurel}{\Longrightarrow} 
\newcommand{\lstau}[2]{#1\lstaurel #2}
\newcommand{\lss}[3]{#1\stackrel{#2}{\lstaurel}#3} 
\newcommand{\Div}[1]{#1\!\!\Uparrow} 
\newcommand{\simone}{\mathcal R}
\newcommand{\bisim}{\approx}

\newcommand{\behone}{\mathfrak b}
\newcommand{\behfun}{\mathfrak f}
\newcommand{\behserv}{\mathfrak s}
\newcommand{\behonline}{\mathfrak o}
\newcommand{\fun}{\mathrm{fun}}
\newcommand{\lang}{\mathrm{lang}}
\newcommand{\Fun}{\mathcal{FUN}}
\newcommand{\Serv}{\mathcal{SERV}}

\newcommand{\EvType}{\mathcal E}
\newcommand{\evtypone}{a} 
\newcommand{\indep}{\scoh}

\newcommand{\Ev}{\mathrm{Ev}}
\newcommand{\runone}{\varphi} 
\newcommand{\runtwo}{\psi} 
\newcommand{\eqrun}{\sim} 
\newcommand{\prerun}{\lesssim} 
\newcommand{\evone}{e} 
\newcommand{\evtwo}{d} 
\newcommand{\twght}{w_t}
\newcommand{\swght}{w_s}
\newcommand{\EvTypeOf}[1]{\mathrm{evtype}(#1)}

\newcommand{\bitone}{b}
\newcommand{\proczero}{\mathbf{0}}
\newcommand{\procone}{P}
\newcommand{\proctwo}{Q}
\newcommand{\procthree}{R}
\newcommand{\procfour}{S}
\newcommand{\pvarone}{A}

\newcommand{\relpara}{\mathrel{|}}
\newcommand{\para}[2]{#1\relpara #2}
\newcommand{\fv}[1]{\mathsf{FV}(#1)}
\newcommand{\ite}[3]{#1.(#2,#3)}
\newcommand{\inp}[2]{#1(#2)}
\newcommand{\inpc}[2]{#1(#2).}
\newcommand{\out}[2]{\overline{#1}\langle #2\rangle}
\newcommand{\outc}[2]{\overline{#1}\langle #2\rangle.}
\newcommand{\evarone}{x}
\newcommand{\evartwo}{y}
\newcommand{\eqz}[1]{\mathtt{0}_{?}(#1)}
\newcommand{\eqe}[1]{\mathit{\varepsilon_{?}}(#1)}
\newcommand{\apz}[1]{\mathtt{0}(#1)}
\newcommand{\apu}[1]{\mathtt{1}(#1)}
\newcommand{\tail}[1]{\mathit{tail}(#1)}
\newcommand{\ichexpone}{I}
\newcommand{\ochexpone}{O}
\newcommand{\expone}{E}
\newcommand{\exptwo}{F}
\newcommand{\bexpone}{B}
\newcommand{\true}{\mathit{tt}}
\newcommand{\false}{\mathit{ff}}
\newcommand{\defeq}{\stackrel{\mathrm{def}}{=}}
\newcommand{\ptag}{p}

\newcommand{\memone}{M}
\newcommand{\qfunone}{\Theta}
\newcommand{\qone}{q}
\newcommand{\msetone}{\Gamma}
\newcommand{\cfg}[2]{[#1]#2}
\newcommand{\cfgone}{C}

\newcommand{\mtrans}[1]{\stackrel{#1}{\longrightarrow}}
\newcommand{\mtautrans}{\mtrans\sa}
\newcommand{\val}[2]{#2^{#1}}
\newcommand{\labsem}[1]{[#1]}
\newcommand{\Ops}{\mathrm{Op}}
\newcommand{\opone}{l}
\newcommand{\walts}[1]{\llbracket #1\rrbracket}

\newcommand{\wght}{\$}
\newcommand{\scost}[1]{\mathrm{space}(#1)}
\newcommand{\tcostmem}[2]{\mathrm{time}_{#1}(#2)}
\newcommand{\tot}[1]{\mathrm{tot}(#1)}
\newcommand{\tcost}[1]{\mathrm{time}(#1)}
\newcommand{\Inp}[1]{\mathrm{Inp}(#1)}
\newcommand{\inpsize}[1]{\|#1\|}

\newcommand{\BTS}[2]{\mathbf{BTS}(#1,#2)}
\newcommand{\Time}{\mathbf{TIME}}
\newcommand{\ATime}{\mathbf{ATIME}}

\newcommand{\FunTS}[2]{\mathbf{FUNTS}(#1,#2)}
\newcommand{\ServTS}[2]{\mathbf{SERVTS}(#1,#2)}

\newcommand{\classP}{\ensuremath{\mathbf P}}
\newcommand{\classPSPACE}{\ensuremath{\mathbf{PSPACE}}}
\newcommand{\classL}{\ensuremath{\mathbf L}}

\newcommand{\classNC}{\ensuremath{\mathbf{NC}}}
\newcommand{\classNP}{\ensuremath{\mathbf{NP}}}
\newcommand{\classIP}{\ensuremath{\mathbf{IP}}}
\newcommand{\encode}[1]{{#1}^\bullet}

\newcommand{\esm}{\textsc{esm}}

\newcounter{theo}
\newcounter{rem}
\newcounter{de}

\newtheorem{definition}[de]{Definition}
\newtheorem{lemma}[theo]{Lemma}

\newtheorem{proposition}[theo]{Proposition}
\newtheorem{corollary}[theo]{Corollary}
\newtheorem{theorem}[theo]{Theorem}

\newtheorem{remark}[rem]{Remark}
\newenvironment{proof}{\begin{trivlist}
       \item[\hskip \labelsep {\bfseries Proof.}]}{\hfill $\Box$ \end{trivlist}}

\newcommand{\qed}{\hfill$\Box$}

\begin{document}


\title{Computational Complexity\\ of Interactive Behaviors}

\author{
Ugo Dal Lago\footnote{Dip.~di Scienze dell'Informazione -- Univ.\ 
di Bologna, Italy
}
\and
Tobias Heindel\footnote{CEA -- LIST, Gif-sur-Yvette,  France
}
\and
Damiano Mazza\footnote{LIPN -- CNRS and Universit\'e Paris 13, France
}
\and
Daniele Varacca\footnote{PPS -- CNRS and Universit\'e Paris Diderot, France
}
}

\date{}

\maketitle

\begin{abstract}
	The theory of computational complexity focuses on functions and, hence, studies programs whose interactive behavior is reduced to a simple question/answer pattern.  We propose a broader theory whose ultimate goal is expressing and analyzing the intrinsic difficulty of fully general interactive behaviors.  To this extent, we use standard tools from concurrency theory, including labelled transition systems (formalizing behaviors) and their asynchronous extension (providing causality information).  Behaviors are implemented by means of a multiprocessor machine executing CCS-like processes. The resulting theory is shown to be consistent with the classical definitions: when we restrict to functional behaviors (i.e., question/answer patterns), we recover several standard computational complexity classes.
\end{abstract}

\section{Introduction}
In the early days, computers were considered as oracles: one would have a question and the computer would provide the answer. For instance,
one day the American army had just launched a rocket to the Moon, and the four star General typed in two questions to the computer: (1) Will the rocket reach the Moon? (2) Will the rocket return to the Earth? The computer did some calculations for some time and then ejected a card which read: ``Yes.'' The General was furious; he didn't know whether ``Yes'' was the answer to the first question, the second or both. Therefore he angrily typed in ``Yes, what?''. The computer did some more calculations and then printed on a card: ``Yes, Sir.''\footnote{Adapted from Raymond Smullyan: \emph{What is the name of this book?}}

That every computation may eventually be reduced to the input/output pattern is an assumption underlying most of classical computability theory. The theory of computational complexity is an excellent example: it studies the intrinsic difficulty of \emph{problems}, which are nothing but ``yes or no'' questions. Accordingly, the classical methods that measure the complexity of a program ignore the possibility that it may interact with its environment between the initial request and the final 
answer; even when a more complex interaction pattern is considered (\eg, in interactive proofs~\cite{IP}), it often is seen as yet another way to solve problems (viz.\ the class \classIP, which is a class of problems).

Nowadays, we live in a world of ubiquitous computing systems that are highly interactive and communicate with their environments, following possibly complicated protocols. These computing systems are fundamentally different from those that just provide answers to questions without any observable intermediate actions. To study this phenomenon of \emph{interactive computation}~\cite{Wegener}, theoretical computer scientists have developed several formalisms and methodologies, such as process calculi and algebras.

However, little has been done so far to tackle the computational complexity of interactive systems (one of the few examples being the competitive analysis of online algorithms~\cite{OnlineAlgos}). Note that, as mentioned above, this issue is beyond the classical theory of computational complexity. 
Thus, we set out to provide grounds for a revised theory of computational complexity that is capable of gauging the efficiency of genuinely interactive behaviors.

The first conceptual step is the formalization of \emph{behaviors}. Our approach follows standard lore of concurrency theory: a behavior is an equivalence class of \emph{labelled transition systems} (\textsc{lts}), which in turn are usually specified using process calculi, such as Milner's \textsc{ccs}~\cite{Milner:Pi}. In this paper, we use a variation of bisimilarity as behavioral equivalence; however, other equivalences that have been proposed in the literature, such as coupled similarity and testing, work equally well, as long as certain minimal requirements are satisfied. Shifting the focus towards behaviors, the fundamental question of classical computational complexity ``What is the cost of solving a problem (or implementing a function)?'' becomes ``What is the cost of implementing a behavior?''.

A suitable cost model is not as easily found as in the functional case where we just measure the resources (time, space) required to compute the answer as a function of the size of the question. Of course, the resources depend on the chosen computational model (such as Turing machines), but the general scheme does not depend on the specific model. We propose a notion of cost for general interactive behaviors that abstracts away from a specific model. Costs are attributed to events in \emph{weighted asynchronous \textsc{lts}} (or \textsc{walts}): asynchrony is a standard feature that is added to transition systems to represent causal dependencies~\cite{WN95}, which we need to generalize the trivial dependency between questions and answers; weights are used to specify additional quantitative information about space and time consumption.

Finally, we introduce a computational model, the \emph{process machine}, which implements behaviors (just as Turing machines implement functions) by executing concurrent programs written in a \textsc{ccs}-based language. Such a machine has an unbounded number of processors each equipped with a private memory and capable of performing basic string manipulation and communicating asynchronously with other processors or the external environment. The process machine admits a natural semantics in terms of \textsc{walts}s and thus provides us with a non-trivial, paradigmatic instance of our abstract framework for measuring the complexity of behaviors.

Complexity classes are then defined as sets of behaviors that can be implemented by a process running within given time and space bounds on the process machine. We conclude by showing that if we restrict to functional behaviors (\ie, trivial input/output patterns) we obtain several standard complexity classes; thus, at least in many paradigmatic cases, we have in fact a consistent extension of complexity theory into the realm of interactive computation. As a further sanity check, we verify that the complexity of a function is invariant under some different (but intuitively equivalent) representations that may be given of it in terms of behaviors.

\section{Behaviors}
\label{sec:Beh}
In this section we formally define behaviors as equivalence classes of labelled transition systems. Such systems can receive messages on some \emph{input channels}, send messages on some \emph{output channels}, and perform internal, invisible, computation. With the aim of being as concrete as possible, we consider messages to be binary strings.
We denote by $\Words=\{0,1\}^\ast$ the set of such strings, with $\emstr$ denoting the empty string. We also fix two disjoint sets $\IChan,\OChan$ of input and output channel names. 

\begin{definition}[Labelled transition system]
	\label{def:LTS}
	An \emph{input action} (resp.\ \emph{output action}) is an element of $\IChan\times\Words$ (resp.\ $\OChan\times\Words$); together, they form the set of \emph{visible actions}, denoted by $\Vis$. The set of \emph{actions} is $\Act=\Vis\cup\{\sa\}$, where $\sa$ is the \emph{internal action}.

	A \emph{labelled transition system} (\textsc{lts} for short) is a triple $\lsysone=(\Web\lsysone,\initone,\trans_\lsysone)$, where $\Web\lsysone$ is a set, whose elements are called \emph{states}, $\initone\in\Web\lsysone$ is the \emph{initial state}, and $\trans_\lsysone\subseteq\Web\lsysone\times\Act\times\Web\lsysone$ is the \emph{transition relation}. 
\end{definition}
Given an \textsc{lts} $\lsysone$, we write $\ls{\stateone}{\acone}{\stateone'}$ when $(\stateone,\acone,\stateone')\in\trans_\lsysone$. 
Since internal computation is invisible, it is standard practice to consider several internal steps as one single, still invisible step.
We denote by $\lstaurel$ the reflexive-transitive closure of $\ls{}{\sa}{}$ and, given $\alpha\in\Vis$, we write $\lss{\stateone}{\alpha}{\statetwo}$ just if there exist $\stateone',\statetwo'$ such that $\stateone\lstaurel\ls{\stateone'}{\alpha}{\statetwo'}\lstaurel\statetwo$. 

The standard notion of equivalence of transition systems is \emph{bisimilarity}.

\begin{definition}[Bisimilarity]
	Let $\lsysone,\lsystwo$ be \textsc{lts}s, with initial states $\initone,\inittwo$, respectively. A \emph{simulation} from $\lsysone$ to $\lsystwo$ is a relation $\simone\subseteq\Web\lsysone\times\Web\lsystwo$ such that $(\initone,\inittwo)\in\simone$ and, for all $(\stateone,\statetwo)\in\simone$, we have:
	\begin{enumerate}
		\item if $\ls{\stateone}{\acone}{\stateone'}$ with $\acone\in\Vis$, then there exists $\statetwo'$ such that $\lss{\statetwo}{\acone}{\statetwo'}$ and $(\stateone',\statetwo')\in\simone$;
		\item if $\ls{\stateone}{\sa}{\stateone'}$, then there exists $\statetwo'$ such that $\lstau{\statetwo}{\statetwo'}$ and $(\stateone',\statetwo')\in\simone$.
	\end{enumerate}
A simulation $\simone$ from $\lsysone$ to $\lsystwo$ is a \emph{bisimulation} if $\op\simone=\{(t,s)\in \Web\lsystwo\times\Web\lsysone \st (s,t)\in \simone\}$ is a simulation from $\lsystwo$ to $\lsysone$. We define $\lsysone\bisim\lsystwo$ iff there exists a bisimulation between $\lsysone$ and $\lsystwo$. This relation is called \emph{bisimilarity}.	
\end{definition}
Bisimilarity can be shown to be an equivalence relation.

For our purposes we furthermore require that the equivalence does not \emph{introduce
divergence}. Given $\stateone\in\Web\lsysone$, we say that \emph{there is a divergence at $\stateone$} (denoted as $\Div\stateone$) if there exists an infinite sequence $\ls\stateone\sa{\ls{\stateone_1}\sa{\ls{\stateone_2}\sa\cdots}}$.
\begin{definition}[Divergence-sensitive bisimilarity]
We say that a (bi)simulation $\simone$ between $\lsysone$ and $\lsystwo$  \emph{does not introduce divergence} if, for all $(\stateone,\statetwo)\in\simone$, $\Div\statetwo$ implies $\Div\stateone$. We define \emph{divergence-sensitive bisimilarity}, denoted by $\bisim_d$, by requiring the existence of a bisimulation not introducing divergence.
\end{definition}

For our purposes, weaker equivalences (such as \emph{coupled simulation} \cite{ParrowSjodin,vanGlabbeek}) suffice, and might actually even be desirable. Whatever equivalence is chosen, the essential point is that it does not introduce divergence.

\begin{definition}[Behavior]
	A \emph{behavior} is a $\bisim_d$-equivalence class.
\end{definition}

In the sequel, it will be useful to have a compact notation for describing \textsc{lts}'s. For this, we shall use a notation similar to the syntax of Milner's \textsc{ccs}~\cite{Milner:Pi}. For instance, if $f:\Words\rightarrow\Words$ is a function, $\inp{\ichan}{x}.\out{\ochan}{f(x)}$ denotes the \textsc{lts} whose states are $\{\initone\}\cup\bigcup_{\strone\in\Words}\{\stateone_\strone,\statetwo_\strone\}$ and whose transitions are $\ls\initone{\inp\ichan\strone}{\stateone_\strone}$ and $\ls{\stateone_\strone}{\out{\ochan}{f(\strone)}}{\statetwo_\strone}$, for all $\strone\in\Words$. This kind of \textsc{lts} is used to define the behaviors that correspond to classical input/output computations.

\begin{definition}[Functional behavior]
	\label{def:FunBeh}
	In the following, we fix two channels $\ichan\in\IChan$ and $\ochan\in\OChan$. Let $f:\Words\rightarrow\Words$ be a function. The \emph{functional behavior} induced by $f$, denoted by~$\behfun_f$, is the equivalence class of $\inpc{\ichan}{x}\out{\ochan}{f(x)}$. We denote by $\Fun$ the set of all functional behaviors.
\end{definition}
\begin{lemma}
	\label{lemma:FunUnique}
	Let $f,g:\Words\rightarrow\Words$. Then, $f=g$ iff $\behfun_f=\behfun_g$.
\end{lemma}
\begin{definition}[Language of a functional behavior]
	\label{def:Lang}
	By \reflemma{FunUnique}, every functional behavior $\behone\in\Fun$ determines a unique function $\fun\behone$ on $\Words$ such that $\behone=\behfun_{\fun\behone}$. This induces a language (\ie, a subset of $\Words$) $\lang\behone=\{\strone\in\Words\st \fun\behone(\strone)=\emstr\}$.
\end{definition}

\section{Abstract Cost Models for Interactive Computation}
\label{sec:Icmplx}
In order to define the complexity of behaviors, we need to add concurrency and causality information to keep track of the dependencies of outputs on relevant, ``previous'' inputs and to identify independent ``threads'' of computation in a parallel algorithm. There are several models of concurrency in the literature (see \cite{WN95} for an overview of standard approaches). \emph{Asynchronous transition systems} \cite{WN95}, which are an extension of the well known model of Mazurkiewicz traces~\cite{Maz86}, are sufficiently expressive for our purposes. In order to speak about complexity, we shall add a notion of \emph{weight}: on transitions, for time complexity, and on states, for space complexity. This justifies our choice of asynchronous transition systems, which have an explicit notion of state, over the \emph{a priori} simpler model of Mazurkiewicz traces.
\begin{definition}[Asynchronous LTS \cite{WN95}]
	\label{def:ALTS}
	An \emph{asynchronous \textsc{lts}} \textsc{(alts)} is a tuple $\lsysone=(\Web\lsysone,\initone,\EvType(\lsysone),\trans_\lsysone,\indep_\lsysone)$ where $\Web\lsysone$ is a set of \emph{states}, $\initone\in\Web\lsysone$ is the \emph{initial state}, $\EvType(\lsysone)$ is a set of \emph{event types}, $\trans_\lsysone\subseteq\Web\lsysone\times\EvType(\lsysone)\times\Web\lsysone$ is the \emph{transition relation} and $\indep_\lsysone$ is an antireflexive, symmetric relation on $\EvType(\lsysone)$, called \emph{independence relation}, such that (using the notations of \refdef{LTS}):
	\begin{enumerate}
		\item $\evtypone\in\EvType(\lsysone)$ implies $\stateone\mtrans{\evtypone}\statetwo$ for some $\stateone,\statetwo\in\Web\lsysone$;
		\item $\ls{\stateone}{\evtypone}{\stateone'}$ and $\ls{\stateone}{\evtypone}{\stateone''}$ implies $\stateone'=\stateone''$;
		\item $\evtypone_1\indep\evtypone_2$ and $\ls{\stateone}{\evtypone_1}{\stateone_1}$, $\ls{\stateone}{\evtypone_2}{\stateone_2}$ implies $\exists\statetwo\in\Web\lsysone$ s.t.\ $\ls{\stateone_1}{\evtypone_2}{\statetwo}$ and $\ls{\stateone_2}{\evtypone_1}{\statetwo}$;
		\item $\evtypone_1\indep\evtypone_2$ and $\ls{\stateone}{\evtypone_1}{\ls{\stateone_1}{\evtypone_2}{\statetwo}}$ implies $\exists\stateone_2\in\Web\lsysone$ s.t.\ $\ls{\stateone}{\evtypone_2}{\ls{\stateone_2}{\evtypone_1}{\statetwo}}$.
	\end{enumerate}
\end{definition}

In complete analogy to the definitions for Mazurkiewicz traces, 
we have trace equivalence classes of transition sequences in \textsc{alts}s 
and we define events with the expected causality relation that relates them. 

\begin{definition}[Run, trace equivalence, event, causal order]
	A \emph{run} in an \textsc{alts} $\lsysone$ is a finite, possibly empty sequence of consecutive transitions $\runone=\stateone\mtrans{\evtypone_1}\cdots\mtrans{\evtypone_n}\statetwo$, which we denote by $\stateone\mtrans{\runone}\statetwo$. Concatenation of runs is denoted by juxtaposition. \emph{Trace equivalence}, denoted by $\eqrun$, is the smallest equivalence relation on runs such that, for all $\evtypone_1\indep\evtypone_2$, if $\runone=\stateone'\mtrans{\runone'}\stateone\mtrans{\evtypone_1}\stateone_1\mtrans{\evtypone_2}\statetwo\mtrans{\runone''}\statetwo'$ and $\runtwo=\stateone'\mtrans{\runone'}\stateone\mtrans{\evtypone_2}\stateone_2\mtrans{\evtypone_1}\statetwo\mtrans{\runone''}\statetwo'$ with $\stateone,\stateone_1,\stateone_2,\statetwo$ as in point 3 of \refdef{ALTS}, then $\runone\eqrun\runtwo$. We define a preorder between runs by $\runone\prerun\runtwo$ iff $\runtwo\sim\runone\runone'$ for some run $\runone'$. A run $\runone$ is \emph{essential} if it is of the form $\initone\mtrans{\runone'}\stateone\mtrans{\evtypone}\statetwo$, with $\initone$ the initial state of $\lsysone$, and for all $\runtwo\eqrun\runone$, we have $\runtwo=\initone\mtrans{\runtwo'}\stateone\mtrans{\evtypone}\statetwo$ with $\runtwo'\eqrun\runone'$.

	An \emph{event} is a $\eqrun$-equivalence class of essential runs. We denote by $\Ev(\lsysone)$ the set of events of $\lsysone$; it is a poset under the quotient relation $\prerun/\!\!\eqrun$, which we denote by $\leq$ and call \emph{causal order}. Note that, if $\evone\in\Ev(\lsysone)$, all $\runone\in\evone$ ``end'' with the same transition; we denote by $\EvTypeOf\evone$ the event type of this transition.
\end{definition}

Finally we can add data for ``time consumption'' of event types and the ``size'' of states, 
which allow to define the time and space cost of events. 

\begin{definition}[Weights]
	A \emph{weighted \textsc{alts}} \textsc{(walts)} is a triple $(\lsysone,\twght,\swght)$, where $\lsysone$ is an \textsc{alts}, $\twght\colon\EvType(\lsysone) \to \Nat$ is the \emph{time weight},  and $\swght \colon \Web{\lsysone} \to \Nat$ is the \emph{space weight}.

	\sloppy{Let $\runone=\stateone_0\mtrans{\evtypone_1}\cdots\mtrans{\evtypone_n}\stateone_n$ be a run. Its \emph{space cost} is $\scost\runone=\max_{0\leq i\leq n}\swght(\stateone_i)$. The \emph{space cost} of an event $\evone\in\Ev(\lsysone)$ is $\scost\evone=\max_{\runone\in\evone}\scost\runone$.}

	Let $\evone\in\Ev(\lsysone)$. We denote by $\tot\evone$ the set of chains of events, \ie, totally ordered subsets of $(\Ev(\lsysone),\leq)$, whose maximum is $\evone$. The \emph{time cost} of $\evone$ is
	$$\tcost{\evone}=\max_{X\in\tot\evone}\sum_{\evtwo\in X}\twght(\EvTypeOf\evtwo).$$
\end{definition}

Roughly speaking,  the space cost of events is independent of their scheduling; 
however, for the time cost of an event we assume an ``ideal'' scheduler that fully exploits all concurrency of the \textsc{walts}.


\section{The Process Machine}
\label{sec:Proc}
We start by defining \emph{string expressions} and \emph{Boolean expressions}, which are generated by the following grammar:
\begin{align*}
	\expone,\exptwo & ::=
	\evarone \gpipe
	\strone \gpipe
	\apz{\expone} \gpipe
	\apu{\expone} \gpipe
	\tail{\expone} \\
	\bexpone & ::=
	\true \gpipe
	\false \gpipe
	\eqz{\expone} \gpipe
	\eqe{\expone},
\end{align*}
where $\evarone$ ranges over a denumerably infinite set of \emph{variables}, and $\strone$ ranges over $\Words$.

\emph{Processes} are defined by the following grammar:
\begin{align*}
	\procone,\proctwo&::=
	\proczero \gpipe
	\pvarone\langle\expone_1,\ldots,\expone_n\rangle \gpipe
	\outc\ochexpone\expone\procone \gpipe
	\inpc\ichexpone\evarone\procone \gpipe
	\ite\bexpone\procone\proctwo \gpipe
	\para{\procone}{\proctwo}.
\end{align*}
where $\ochexpone$ stands for either an output channel $\ochan\in\OChan$ or a string expression, $\ichexpone$ stands for either an input channel $\ichan\in\IChan$ or a string expression, $\expone,\expone_1,\ldots,\expone_n$ range over string expressions, and $\pvarone$ ranges over a denumerably infinite set of \emph{process identifiers}, each coming with an \emph{arity} $n\in\Nat$ and a \emph{defining equation} of the form
$$\pvarone(\evarone_1,\ldots,\evarone_n)\defeq\procone$$
where $P$ is a process whose free variables are included in $\evarone_1,\ldots,\evarone_n$. As usual in process calculi, the free variables of a process (denoted by $\fv\procone$) are defined to be the variables not in the scope of an input prefix $\inp\ichexpone\evarone$, which \emph{binds} $\evarone$. A process $\procone$ is \emph{closed} if $\fv\procone=\emptyset$. In the following, all bound variables of a process are supposed to be pairwise distinct.

\begin{table}[t]
	\begin{displaymath}
		\begin{array}{rrcl}
		\mathbf{Nil:} &
			\cfg{(\proczero,\memone)_\ptag,\msetone}\qfunone
			&\quad\mtautrans\quad&
			\cfg{\msetone}\qfunone
		\\
		\mathbf{Rec:} &
			\cfg{(\pvarone\langle\expone_1,\ldots\expone_n\rangle,M)_\ptag,\msetone}{\qfunone}
			&\mtautrans&
			\cfg{(\procone,\{ \evarone_1\mapsto\val\memone{\expone_1},\ldots,\evarone_n\mapsto\val\memone{\expone_n} \})_\ptag,\msetone}{\qfunone}
		\\
		&&& \textrm{with }\pvarone(\evarone_1,\ldots,\evarone_n)\defeq P
		\\
		\mathbf{Snd:} &
			\cfg{(\outc{\expone}{\exptwo}\procone,\memone)_\ptag,\msetone}{\qfunone}
			&\mtautrans&
			\cfg{(\procone,\memone)_\ptag,\msetone}{\qfunone'}
		\\
		&&& \textrm{with }\qfunone'(\val\memone\expone)=\qfunone(\val\memone\expone)\cdot\val\memone\exptwo,
		\\
		&&& \textrm{and }\qfunone'=\qfunone\textrm{ everywhere else}
		\\
		\mathbf{Rcv:} &
			\cfg{(\inpc{\expone}{\evarone}\procone,\memone)_\ptag,\msetone}{\qfunone}
			&\mtautrans&
			\cfg{(\procone,\memone\cup\{ \evarone\mapsto\strone \})_\ptag,\msetone}{\qfunone'}
		\\
		&&& \textrm{only if }\qfunone(\val\memone\expone)=\strone\cdot\qone.\textrm{ Then, }\qfunone'(\val\memone\expone)=\qone
		\\
		&&& \textrm{and }\qfunone'=\qfunone\textrm{ everywhere else}
		\\
		\mathbf{Out:} &
			\cfg{(\outc{\ochan}{\expone}\procone,\memone)_\ptag,\msetone}{\qfunone}
			&\mtrans{\out{\ochan}{\val\memone\expone}}&
			\cfg{(\procone,\memone)_\ptag,\msetone}{\qfunone}
		\\
		\mathbf{Inp:} &
			\cfg{(\inpc{\ichan}{\evarone}\procone,\memone)_\ptag,\msetone}{\qfunone}
			&\mtrans{\inp{\ichan}{\strone}}&
			\cfg{(\procone,\memone\cup\{ \evarone\mapsto\strone \})_\ptag,\msetone}{\qfunone}
		\\
		\mathbf{Cnd:} &
			\cfg{(\ite{\bexpone}{\procone}{\proctwo},\memone)_\ptag,\msetone}{\qfunone}
			&\mtautrans&
			\left\{\begin{array}{ll}
				\cfg{(\procone,\memone)_\ptag,\msetone}{\qfunone} & \textrm{if }\val\memone\bexpone=\true, \\
				\cfg{(\proctwo,\memone)_\ptag,\msetone}{\qfunone} & \textrm{if }\val\memone\bexpone=\false
			\end{array}\right.
		\\
		\mathbf{Spn:} &
			\cfg{(\para{\procone}{\proctwo},\memone)_\ptag,\msetone}{\qfunone}
			&\mtautrans&
			\cfg{(\procone,\memone)_{\ptag0},(\proctwo,\memone)_{\ptag1},\msetone}{\qfunone}
		\end{array}
	\end{displaymath}
	\caption{The transitions of the process machine.}
	\label{tab:ProcMachine}
\end{table}

To assign values to expressions, we use \emph{environments}, \ie, finite partial functions from variables to $\Words$. 
If $\expone$ is a string expression whose variables are all in the domain of an environment $\memone$, we define its \emph{value} $\val\memone\expone$ by induction: $\val\memone\evarone=M(\evarone)$; $\val\memone\strone=\strone$; $\val\memone{\apz{\expone}}=0\val\memone\expone$; $\val\memone{\apu{\expone}}=1\val\memone\expone$; and $\val\memone{\tail{\expone}}=\strone$ if $\val\memone\expone=\bitone\strone$, with $\bitone\in\{0,1\}$. Similarly, we define the \emph{value} of Boolean expressions: $\val\memone\true=\true$; $\val\memone\false=\false$; $\val\memone{\eqz{\expone}}=\true$ if $\val\memone\expone=0\strone$, otherwise it is $\false$; and $\val\memone{\eqe{\expone}}=\true$ if $\val\memone\expone=\emstr$, otherwise it is $\false$.

\begin{definition}[Machine configurations, transitions]
	\label{def:Machine}
	\sloppy{A \emph{processor state} is a triple $(\procone,\memone)_\ptag$ where $\procone$ is a process, $\memone$ is an environment whose domain includes $\fv\procone$, and $\ptag$ is a binary string, the \emph{processor tag}.}

	A \emph{queue function} is a function $\qfunone$ from $\Words$ to finite lists of $\Words$, which is almost everywhere equal to the empty list. In the following, lists of words are ranged over by $\qone$, and we denote by $\cdot$ their concatenation.

	A \emph{configuration} $\cfgone$ is a pair $\cfg\msetone\qfunone$, where $\msetone$ is a set of processor states whose processor tags are pairwise incompatible in the prefix order (\ie, no processor tag is the prefix of another), and $\qfunone$ is a queue function. 

\end{definition}

\begin{definition}[LTS of a process]
	Let $\procone$ be a closed process. We define $\labsem\procone$ to be the \textsc{lts} generated by \reftab{ProcMachine} with the initial state $\cfg{(\procone,\emptyset)_\emstr}{\epsilon}$ (empty environment, tag and queue function).
\end{definition}

The reader acquainted with process algebras will note how, in spite of the presence of output prefixes in the syntax of processes, the machine treats outputs asynchronously: strings are sent (internally or externally) without waiting to synchronize with a receiver.

Given a deterministic Turing machine computing the function $f:\Words\rightarrow\Words$, it is possible to exhibit a closed process $\procone$ such that $\labsem{\procone}\in\behfun_f$; moreover, the execution of this process on the machine uses only one processor. Many more standard, ``functional'' models of computation can be simulated by our process machine (see \ref{sec:Encodings}). However, the process machine is obviously richer, in the sense that it may implement more complex, ``non-functional'' interactive behaviors.

As announced, each transition will be given a weight, and each configuration a size. For every string expression $\expone$ and Boolean expression $\bexpone$, given an environment $\memone$ whose domain contains the variables of $\expone$ and $\bexpone$, we fix positive integers $\tcostmem\memone\expone$ and $\tcostmem\memone\bexpone$, representing the time it takes for a processor with environment $\memone$ to compute the string $\val\memone\expone$ and the Boolean $\val\memone\bexpone$. In the following we denote by $\len\strone$ the length of $\strone\in\Words$, and the size of an environment $\memone$ is $\len{\memone}=\sum_{\evarone\in\dom(\memone)}(\len{\memone(\evarone)}+1)$.

\begin{definition}[Weight of transitions and size of configurations]
	 The \emph{weight} of a machine transition $t$, denoted by $\wght t$, is defined as follows, with reference to \reftab{ProcMachine}:

	\begin{tabular}{rlcrl}
		$\mathbf{Nil:}$ & $\wght t=1$ & \quad\ & $\mathbf{Out:}$ & $\wght t=1+\tcostmem{\memone}{\expone}$ \\
		$\mathbf{Rec:}$ & $\wght t=1+\sum_{i=1}^n\tcostmem{\memone}{\expone_i}$ && $\mathbf{Inp:}$ & $\wght t=1+\len\strone$ \\
		$\mathbf{Snd:}$ & $\wght t=1+\tcostmem{\memone}{\expone}+\tcostmem{\memone}{\exptwo}$ && $\mathbf{Cnd:}$ & $\wght t=1+\tcostmem{\memone}{\bexpone}$ \\
		$\mathbf{Rcv:}$ & $\wght t=1+\tcostmem{\memone}{\expone}+\len{\strone}$ && $\mathbf{Spn:}$ & $\wght t=1+\len\memone$
	\end{tabular}

	If $q$ is a list of strings, its size $\len\qone$ is the sum of the lengths of the strings appearing in $\qone$; then, the size of a queue function $\qfunone$ is $\len\qfunone=\sum_{\strone\in\dom(\qfunone)}\len{\qfunone(\strone)}$. Finally, the \emph{size} of a configuration $\cfgone=\cfg{(\procone_1,\memone_1),\ldots,(\procone_n,\memone_n)}{\qfunone}$ is $\len\cfgone=\len{\qfunone}+\sum_{i=1}^n\len{\memone_i}$.
\end{definition}
\begin{definition}[WALTS of a process]
	We define the set of \emph{operations} as
	$\Ops=\{\mathbf{Nil}$, $\mathbf{Rec}$, $\mathbf{Snd}$, $\mathbf{Rcv}$, $\mathbf{Out}$, $\mathbf{Inp}$, $\mathbf{Cnd}$, $\mathbf{Spn} \}$. Let $\procone$ be a closed process. We define a \textsc{walts} $\walts\procone$ as follows:
	\begin{varitemize}
		\item $\Web{\walts\procone}=\Web{\labsem{\procone}}$;
		\item the initial state is $\cfg{(P,\emptyset)_\emstr}{\epsilon}$;
		\item $\EvType(\walts{\procone})$ is the set of all $(\ptag,\opone,\natone)\in\Words\times\Ops\times\Nat$ s.t.\ in $\labsem{\procone}$ there is a transition $t$ of type $\opone$ performed by a processor whose tag is $\ptag$ and s.t.\ $\wght t=\natone$;
		\item the independence relation is the smallest symmetric relation s.t.\ $(\ptag,\opone,\natone)\indep(\ptag',\opone',\natone')$ holds as soon as $\ptag\neq\ptag'$ and one of the following conditions is met:
		\begin{varitemize}
			\item $\opone\not\in\{\mathbf{Snd}$, $\mathbf{Rcv}$, $\mathbf{Out}$, $\mathbf{Inp}\}$;
			\item $\opone\in\{\mathbf{Snd}$, $\mathbf{Rcv}\}$ and $\opone'\in\{\mathbf{Out}$, $\mathbf{Inp}\}$;
			\item $\opone,\opone'\in\{\mathbf{Snd}$, $\mathbf{Rcv}\}$ and the transitions concern different queues;
			\item $\opone,\opone'\in\{\mathbf{Out}$, $\mathbf{Inp}\}$ and either $\opone\neq\opone'$ or the transitions concern different external channels.
		\end{varitemize}
		\item $\trans_{\walts{\procone}}=\{(\cfgone,(\ptag,\opone,\natone),\cfgone')\st \forall (\cfgone,\acone,\cfgone')\in\trans_{\labsem{\procone}}$ performed by processor $\ptag$ of type $\opone$ and weight $\natone\}$;
		\item the time weight is $\twght((\ptag,\opone,\natone))=\natone$, and the space weight is $\swght(\cfgone)=\len\cfgone$.
	\end{varitemize}
\end{definition}

Note that two $\mathbf{Snd}$/$\mathbf{Rcv}$ transitions on the same queue are never independent. This amounts to forbidding concurrent access to a queue, even when this could be safe. We could consider queues with concurrent access  at the price of some technical complications.
In this extended abstract, we prefer not to address such an arguably minor detail.

\section{Complexity Classes}
\label{sec:Compl}
We now propose our definition of complexity classes of behaviors. We essentially measure the cost of producing an output as a function
of all the inputs that are below it in the causal order.

\begin{definition}[Input and output events, input size]
	Let $\procone$ be a closed process. An \emph{input event} (resp.\ \emph{output event}) of $\walts{\procone}$ is an event $\evtwo\in\Ev(\walts{\procone})$ s.t.\ $\EvTypeOf{\evtwo}$ is an input (resp.\ output) on an external channel. In the input case, if the string read is $\strone$, we set $\len\evtwo=\len\strone+1$.
	\label{def:InputSize}
	Let $\evone$ be an output event, and let $\Inp\evone$ be the set of input events below $\evone$ (w.r.t.\ the causal order). We define the \emph{input size} of $\evone$ as $\inpsize{\evone}=\sum_{\evtwo\in\Inp{\evone}}\len\evtwo$.
\end{definition}

\begin{definition}[Cost of a process]
	\label{def:ProcCost}
	Let $f,g:\Nat\rightarrow\Nat$. We say that $\procone$ \emph{works in time $f$ and space $g$} 
        if for every output event $\evone$ of $\walts\procone$, $\tcost{\evone}\leq f(\inpsize{\evone})$ and $\scost{\evone}\leq g(\inpsize{\evone})$.
\end{definition}

\begin{definition}[Complexity class]
	\label{def:ComplClass}
	Let $f,g:\Nat\rightarrow\Nat$. We define $\BTS f g$ to be the set of behaviors $\behone$ such that
        there exists a process $\procone$ such that $\labsem{\procone}\in\behone$ and $\procone$ works in time $f$ and
        space $g$.
\end{definition}

As sanity check we show that, in the case of functional behaviors, we essentially recover the standard complexity classes.

\begin{definition}[Functional complexity]
	Let $f,g:\Nat\rightarrow\Nat$. We define the set of languages $\FunTS{f}{g}=\lang(\BTS{f}{g}\cap\Fun)$.
\end{definition}

In the following, $\Time(f)$ and $\ATime(f)$ denote the standard time complexity classes (languages decidable by a deterministic and alternating Turing machine in at most $f(n)$ steps, respectively).
\begin{theorem}
	\label{th:FunClasses}
	Let $f,g:\NN\rightarrow\NN$.
	\begin{enumerate}
		\item $\Time(f(n))\subseteq\FunTS{f(n)}{f(n)}$;
		\item $\FunTS{f(n)}{g(n)}\subseteq\Time(\bigo(f(n)g(n)^h))$ for a constant integer $h>0$;
		\item $\ATime(f(n))\subseteq\FunTS{f(n)}{2^{\bigo(f(n))}}$.
	\end{enumerate}
\end{theorem}
\begin{proof}
	Points (1) and (3) are proved by efficiently encoding Turing machines and alternating Turing machines in the process machine (see \ref{sec:Turing} and \ref{sec:AltTuring}). For point (2), we simulate with a deterministic Turing machine the execution of a process $\procone$ implementing a functional behavior. This is possible because, by the properties of $\bisim_d$, the non-determinism that may be present during the execution of $\procone$ is actually vacuous: when facing a configuration with more than one active processor, the Turing machine may simulate any one of them, without worrying about influencing the outcome or falling into infinite computations. Simulating a single transition of the process machine may be assumed to require at most $c\cdot g(n)^{h'}$ Turing machine steps, where $c,h'$ are constant. Now, a simple combinatorial argument based on the maximum length of runs (which is $f(n)$) and the maximum number of active processors (which is $g(n)$) gives that the Turing machine halts after simulating at most $f(n)g(n)$ transitions, yielding the desired bound. The details are given in \ref{sec:MainThm}.\qed
\end{proof}

\begin{corollary}
	\label{cor:FunClasses}
	Every standard polynomial or superpolynomial deterministic complexity class may be reformulated in terms of $\FunTS{f}{g}$. For instance:
	\begin{displaymath}
		\classP = \bigcup_{k<\omega}\FunTS{n^k}{n^k},\qquad
		\mathbf{EXP} = \bigcup_{k<\omega}\FunTS{2^{n^k}}{2^{n^k}}.
	\end{displaymath}
\end{corollary}

Thanks to the well know equality $\classPSPACE=\mathbf{AP}$, \refth{FunClasses} also immediately implies the inclusion
$\classPSPACE \subseteq \bigcup_{k<\omega}\FunTS{n^k}{2^{n^k}}$ which shows, for instance, that \classNP-complete problems may be solved in polynomial time if we allow an exponential number of processors working in parallel, as expected.


\section{Some Simple Extra-Functional Behaviors}
\label{sec:ExtraFun}
Functional behaviors are only one possible way (albeit perhaps the most natural) of representing functions as behaviors. We analyze here two alternative representations, argue that the complexity of functions should not be altered by switching to these representations, and show that this is indeed the case in our framework, therefore providing a further sanity check.

\begin{definition}[Functional server]
	\label{def:FunServBeh}
	Let $f:\Words\rightarrow\Words$ be a function. The \emph{functional server behavior} induced by $f$, denoted by~$\behserv_f$, is the equivalence class of the \textsc{lts} given by the recursive definition $X\defeq\inpc{\ichan}{x}(\para{\out{\ochan}{f(x)}}{X})$. We denote by $\Serv$ the set of all functional server behaviors.

	A result similar to \reflemma{FunUnique} holds, which allows us to speak of the language induced by a functional server behavior $\behserv$, denoted by $\lang\behserv$, and to define, given functions on natural numbers $f,g$, the class $\ServTS f g=\lang(\BTS f g\cap\Serv)$.
\end{definition}

It is intuitively obvious that a function may be implemented ``repeatedly'' with given time and space bounds iff it may implemented once, with the same bounds.
\begin{theorem}
	\label{th:Serv}
	Let $f,g:\Nat\rightarrow\Nat$.
	\begin{enumerate}
		\item $\ServTS{f}{g}\subseteq\FunTS{f}{g}$;
		\item assuming $f=\Omega(n)$ and monotonic, $\FunTS{f}{g}\subseteq\ServTS{\bigo(f(n))}{g}$.
	\end{enumerate}
\end{theorem}
\begin{proof}
	Let $h:\Words\rightarrow\Words$. Point 1 is proved by taking a process implementing $\behserv_h$ and extracting from it a process whose execution stops after the answer to the first input. For point 2, we take an implementation $P$ of the functional behavior $\behfun_h$ and transform it into a server process which reads an external input, spawns a copy of $P$ to compute $h$ and waits for a further input. For producing the $n$-th output, space is reused (so the space cost is unchanged), while the spawning and successive inputs induce a linear slowdown, which is absorbed in the big-O notation. See \ref{sec:ServThm} for more details.\qed
\end{proof}

In what follows, we endow the set $\Words$ with the prefix order, denoted by $\leq$. A computable monotonic function $h$ may obviously be implemented in the usual, ``offline'' way: a string $\strone$ is given, and $h(\strone)$ is output. However, $h$ may also be computed ``online'': bits are given one at a time and, for each new bit, only the ``difference'' with respect to the output already produced is given. Intuitively, the intrinsic difficulty of computing $h$ should not depend on which of the two implementations is chosen.
\begin{definition}[Online monotonic function]
	Let $h$ be a monotonic function. The \emph{online behavior} induced by $h$, denoted by $\behonline_h$, is the equivalence class of the \textsc{lts} $X\langle\emstr,\emstr\rangle$, with $X$ given by the following recursive definition:
	\begin{align*}
		X(s,r) &\defeq \para{\out{o}{h(s)\setminus r}}{\inpc{i}{x}\mathtt{if}\ [x=\emstr]\ \mathtt{then}\ X\langle s0,h(s)\rangle\ \mathtt{else}\ X\langle s1,h(s)\rangle},
	\end{align*}
	where $\strone\setminus\strtwo$ is defined when $\strtwo\leq\strone$ and is equal to the string $\strtwo'$ such that $\strone=\strtwo\strtwo'$. Note that we take the empty string to represent the bit $0$, and any other string to represent $1$.
%
\end{definition}
\begin{theorem}
	Let $f,g:\Nat\rightarrow\Nat$ be both $\Omega(n)$ and monotonic, and let $h$ be a monotonic function on $\Words$. Then:
	\begin{enumerate}
		\item $\behonline_h\in\BTS f g$ implies $\behfun_h\in\BTS{\bigo(n\cdot f(2n))}{\bigo(g(2n))}$;
		\item $\behfun_h\in\BTS f g$ implies $\behonline_h\in\BTS{\bigo(n\cdot f(n/2))}{\bigo(g(n/2))}$.
	\end{enumerate}
\end{theorem}
\begin{proof}
	Consider an implementation $Q$ of $\behonline_h$. To implement $\behfun_h$, we read the input string and send it bit by bit to $Q$, concatenating the outputs. For a string of size $n$, $Q$ is called $n$ times; the $j$-th time cost is $f(2j)$ (the size of each bit is at most $2$---remember that even the empty string has size $1$), so the total time cost of the final (and only) output event is bounded by $\sum_{j=1}^n f(2j)$. Since $f$ is monotonic, we may bound this by $n\cdot f(2n)$. There is also the cost of sending $n$ bits to $Q$, which is $2n$ and is thus absorbed in the big-O notation. For space, apart from the memory used by $Q$, we need only store the input string and use one bit at a time in the communication channels, so we still get $\bigo(g(n))$.

	Let now $P$ be an implementation of $\behfun_h$. We may implement $\behonline_h$ as follows: we start by sending $\emstr$ to $P$ in order to output $h(\emstr)$; then, each time we get a new bit, we call $P$ on the string received so far and we output only the difference with respect to the output cumulated so far. For the $n$-th output event $o_n$, $P$ is called $n$ times (ignoring the call on the empty string, which has a constant cost); the $j$-th time, we send a string of length $j$ (which costs us $j$) and we receive something after $f(j)$ steps (the time it takes for $P$ to do its computation). The time we need to compute the difference w.r.t. the previous output may be absorbed into $f(j)$, because it is linear. Therefore, we may write $\tcost{o_n}\leq\sum_{j=1}^n(j+f(j))=\bigo(n\cdot f(n)+n^2)=\bigo(n\cdot f(n))$, because $f(n)=\Omega(n)$. Since $\inpsize{o_n}=2n$ (remember that each bit is of size $2$), we have the $1/2$ factor which appears in the statement of the theorem. For space, apart from the space $g(n/2)$ used by $P$, which is reused each time, we need to keep track of the output string (so we can compute the difference), whose size is also bounded by $g(n/2)$.\qed
\end{proof}

A monotonic function $h:\Words\rightarrow\Words$ is said to be \emph{eventually strictly monotonic} (\emph{\esm}) if, for all $n\in\Nat$, there exist $\strone,\strtwo\in\Words$ such that $\len\strone\geq n$, $\strone<\strtwo$ and $h(\strone)<h(\strtwo)$. \esm\ functions may be seen as functions from streams of bits to streams of bits. The online process we introduced above is then an intuitive way of defining a function on streams.

Another point of view \cite{TuckerZucker} sees stream functions as having type $(\Nat\to\{0,1\})\to \Nat\to\{0,1\}$, which can be Curryed into $\Nat\times(\Nat\to\{0,1\})\to \{0,1\}$. That is, given a desired position $i$ on the output stream and an oracle for the input stream, we should be able to say what is the bit at position $i$. For computable functions, the oracle is consulted only finitely many times, and therefore the function is continuous. Thus, a computable stream function corresponds to an \esm\ function on finite strings: a string of length $n$ is mapped to the longest prefix of the stream that does not ask the oracle for the value of the input stream at positions greater than $n$. To capture this representation of streams in our framework, the notion of behavior must be modified to take into account the use of oracles. We plan to do this in future work.

\section{Discussion}
\label{sec:Disc}
The idea to revise and extend the theory of computability (and formal languages)
by replacing functions with behaviors is not new~\cite{GoldinSW01,LeeuwenW00,DBLP:conf/icdcit/BaetenLT11,BaetenLuttikTilburg}.
As an example, the latter works introduce reactive Turing machines (\textsc{rtm}s), 
which are ordinary Turing machines with an additional action (\ie, an element of $\Act$ as in \refdef{LTS})
for each transition between configurations.
Each such \textsc{rtm} induces an \textsc{lts}, 
which then is \emph{executable} by definition; 
two \textsc{rtm}s execute the same behavior
if their \textsc{lts}s are related by a certain behavioral equivalence.
Finally, so-called \emph{effective} \textsc{lts}s, 
i.e.\ \textsc{lts}s with recursively enumerable transition relations, 
coincide with executable ones (up to behavioral equivalence). 
As one might expect, 
it is easy to construct for each \textsc{rtm} (without final states) a corresponding process
such that their \textsc{lts}s are weakly bisimilar (see Appendix~\ref{sec:RTM}).

Not much has been said about interactive \emph{complexity}, however. A notable
exception is Japaridze's system of Clarithmetic~\cite{Japaridze11}, whose focus however is on
logic rather than complexity theory. Another example are lineage of automata \cite{VerbaanLW04}, whose
nature is very finitistic contrarily to the one of our model.



Concerning our own work, in this extended abstract we described merely the first steps of a proposal which, at least in the case of polynomial and superpolynomial deterministic time complexity classes, has the good taste of not being inconsistent with the standard definitions. Starting from here, we have of course a great number of open questions and directions for further investigation.

First of all, in light of \refth{FunClasses}, we may ask how standard \emph{space} complexity classes (\eg\ \classPSPACE) may be recovered from our definitions. 
In this respect, we already know that, in perfect analogy with one-tape Turing machines, the sequential treatment of input strings in the current definition of the process machine prevents us from capturing ``low'' complexity classes, such as \classL\ (or \classNC). To deal with these, random access to the bits of an input string must be allowed (see Appendix~\ref{sec:Boolean} and \ref{sec:PRAM}). 

And then, of course, there is a plethora of questions regarding non-functional behaviors, the main motivation behind our work. What happens when we consider more than one external input/output? What interesting classes of non-functional behaviors can we describe? One issue revealed by our preliminary investigations is that the equivalence chosen in the definition of behavior may need to be changed. For instance, for dealing with streams ({\it cf.} end of \refsect{ExtraFun}), something like \emph{refinement} (\ie, bisimilarity in which some internal choices may be disregarded) seems to be more adapted.


\bibliographystyle{plain}
\bibliography{Biblio}

\newpage
\appendix
\section{Encoding standard computational models}
\label{sec:Encodings}
\subsection{Turing machines}
\label{sec:Turing}
We show how the process machine can simulate a deterministic Turing machine with a constant slowdown.

Let $\delta$ be a function from $\Words$ to processes with free variables among $\overrightarrow x$, 
whose domain is \emph{finite}. Such function may be represented by a process (identifier)
$F_\delta(s,\overrightarrow x)$ which progressively ``explores'' $s$ and returns the appropriate process, 
returning some default process (for instance $\proczero$) in case $s$ does not belong to the domain of $\delta$. 
For example, if $\delta$ is defined only on $\emstr$, $0$ and $1$, and is a closed process in all cases, we have
$$
	F_\delta(s)\defeq\ite{\eqe s}
	{\delta(\emstr)}
	{\ite{\eqz s}
		{\ite{\eqe{\tail{s}}}
			{\delta(0)}
			{\proczero}
		}
		{\ite{\eqe{\tail{s}}}
			{\delta(1)}
			{\proczero}
		}
	}.
$$
Let now $M$ be a deterministic Turing machine with alphabet $\{0,1\}$ and one semi-infinite tape (\ie, 
it is finite to the ``left'' and extends indefinitely to the ``right''). We assume that the input string 
is written in the leftmost cells of the tape, the rest of the tape being covered by blank symbols. We also 
assume that the machine halts with a failure if it attempts to read to the left of the leftmost cell.

The configurations of the machine may be represented by three binary strings $s,l,r$, representing the 
current state, the contents to the left of the head, in reverse order and the contents of the tape 
to the right of the head (including the head as the first symbol of $r$). Then, the transition function of $M$ 
induces three functions of $s$, depending on whether $r$ is empty (the head ``wandered off'' to the right),
starts with a $0$, or starts with a $1$. These, in turn, induce three functions $\delta_\emstr$, $\delta_0$, 
$\delta_1$, from binary strings to processes, all of finite domain, which we describe as follows: let 
$T$ be a process identifier of arity $3$. Suppose the state of $M$ is $s$, and that the current symbol 
is $b$ (which may be blank). Suppose that, from this information, the transition function of $M$ goes to 
state $q$, writes $c$, and moves to the right. Then, we have
$$
\delta_b(s)=T\langle q,c(l),\tail{r}\rangle.
$$
Had the transition function of $M$ decreed a movement to the left instead, we would have
$$
\delta_b(s)=\ite{\eqe l}{\proczero}{\ite{\eqz l}{T\langle q,\tail l,\apz{c(\tail{r})}\rangle}{T\langle q,\tail l,\apu{c(\tail{r})}\rangle}}.
$$
Finally, if $s$ is a halting state, the machine may output the result, which we stipulate to be written to the right of the head 
(including the current position), so we have
$\delta_b(s)=\out{\ochan}{r}$. 

But we still have to define the behavior of the process identifier $T$: the defining equation for it is
$$
T(s,r,l)\defeq\ite{\eqe r}{F_{\delta_\emstr}\langle s,l,r\rangle}{\ite{\eqz r}{F_{\delta_0}\langle s,l,r\rangle}{F_{\delta_1}\langle s,l,r\rangle}}.
$$
The Turing machine $M$ may be represented by the process
$$
\inpc{\ichan}{\evarone}T\langle\strone_0,\emstr,\evarone\rangle,
$$
where $\strone_0$ is the initial state of $M$.

Note how the parallel operator is never used by processes representing Turing machines. This implies that, when such 
processes are executed on the process machine, only one processor is used, and the causal structure of events is 
purely sequential. Note also that we may encode in the same way any computational model based on states whose 
transitions are described by a function of finite domain.

\subsection{Alternating Turing machines}
\label{sec:AltTuring}
It is immediate to define processes corresponding to binary logical operators:
\begin{align*}
  \mathit{And}(a,b,c) &= \inpc a x\inpc b y\ite{\eqz x}{\out c 0}{\ite{\eqz y}{\out c 0}{\out c 1}};\\
  \mathit{Or}(a,b,c) &= \inpc a x\inpc b y\ite{\eqz x}{\ite{\eqz y}{\out c 0}{\out c 1}}{\out c 1}.
\end{align*}
Consider now an alternating Turing machine $M$. We suppose that at every step, $M$ non-deterministically branches 
in two computations; some states will be existential (\ie, will accept if the result of \emph{one} of the two 
branches is accepting), while others will be universal (\ie, will accept if the result of \emph{both} of the 
two branches is accepting). Then, the transition function of $M$ induces 6 functions of finite domain from strings 
to processes, which we call $\delta_b^i(s)$, with $b\in\{0,1,\emstr\}$ and $i\in\{0,1\}$. Intuitively, 
$\delta_b^i(s)$ corresponds to the behavior of $M$ when, given that the current state and symbol are $s$ and $b$, 
the branch $i$ is chosen.

The functions $\delta_b^i(s)$ are defined much like in the case of deterministic Turing machines, except that 
now we use another process identifier $N$, of arity $4$. For instance, if state $s$, symbol $b$, and branch 
$i$ give new state $q$, new symbol $c$, and movement to the right, we have 
$$
\delta_b^i(s)=N\langle q,c(l),\tail{r},d\rangle,
$$ 
and so on. An important difference with deterministic Turing machines is that, in case 
$s$ is an accepting state, we set $\delta_b^i(s)=\out d 1$, and in case it is a rejecting state, we set
$\delta_b^i(s)=\out d 0$, that is, the final decision (accept/reject) is output on an internal channel 
$d$, which is a parameter of $N$, instead of the external output channel $\ochan$.

Then, we introduce two further process identifiers $T_0,T_1$, both of arity $4$, and define them mutually recursively with $N$:
\begin{align*}
  T_i(s,l,r,d) &\defeq \ite{\eqe r}{F_{\delta_\emstr^i}\langle s,l,r,d\rangle}{\ite{\eqz r}{F_{\delta_0^i}\langle s,l,r,d\rangle}{F_{\delta_1^i}\langle s,l,r,d\rangle}}; \\
  N(s,r,l,d) &\defeq \para{\para{T_0\langle s,l,r,\apz d\rangle}{T_1\langle s,l,r,\apu d\rangle}}{\inpc{(\apz d)}{y}\inpc{(\apu d)}{z}F_{Op}\langle s,y,z,d\rangle}.
\end{align*}
where $i\in\{0,1\}$ and $Op(s)$ is the finite-domain function yielding $\mathit{Or}(y,z,d)$ or $\mathit{And}(y,z,d)$ 
according to whether $s$ is an existential or universal state, respectively.

At this point, the alternating Turing machine $M$ may be represented by the process
$$\inpc\ichan\evarone(\para{N\langle \strone_0,\evarone,\emstr,\emstr\rangle}{\inpc{\strone}{\evartwo}\out\ochan\evartwo}),$$
where $\strone_0$ is the initial state.

The reader may check that, upon reception of a string $x$ on the external input channel $\ichan$, the above 
process starts unfolding a parallel computation whose structure is a binary tree of depth proportional to the depth 
of the computation of $M$. Each branch in the tree executes independently from the others; once a leaf is reached, 
the result (acceptance/rejection) is communicated to the parent, which computes a disjunction/conjunction of the 
two data received from its siblings, depending on its existential/universal nature, and passes the result to its 
parent, and so on. The last Boolean computed, which is the final answer of $M$ for accepting or rejecting $x$, is 
sent on channel $\strone$, and is forwarded to the external world through the output channel $\ochan$.

\subsection{Random Access Machines}
A memory cell may be represented by a process which waits on a channel $\strone$ for a string 
$\strtwo$ which is interpreted as follows:
\begin{varitemize}
\item 
  if $\strtwo=\emstr$, no action is taken;
\item 
  the string $\strtwo=0\strtwo'$ is interpreted as a read request, and the value stored in 
  the cell is sent using channel $\strtwo'$;
\item 
  the string $\strtwo=1\strtwo'$ is interpreted as a write request, so the value $\strtwo'$ 
  replaces the current value.
\end{varitemize}
The above process is realized by the following recursive definition:
$$
C(x,v)\defeq \inpc x y\ite{\eqe y}{C\langle x,v\rangle}{\ite{\eqz y}{\outc{\tail{y}}{v}C\langle x,v\rangle}{C\langle x,\tail{y}\rangle}}.
$$

In RAMs, memory cells contain integers, and instructions too refer to integers. Here, we use a unary 
representation: $n$ is represented by the string $0^n$.

A random access memory made of infinitely many cells initially containing zero, located at addresses 
of the form $0^n$, with $n>0$, is generated by the process $M\langle 0\rangle$, where the unary process 
identifier $M$ has the following defining equation $M(c) \defeq \para{C\langle c,\emstr\rangle}{M\langle\apz c\rangle}$. However, such a process is divergent, so we cannot use it for implementing functional behaviors according to \refdef{FunBeh}. Then, we must define a process that only creates a finite number of memory cells at a time, as needed. We first give a couple of auxiliary definitions:
\begin{align*}
	D[\procone,\proctwo](m,n) &\defeq \ite{\eqe{m}}{\procone}{\ite{\eqe{n}}{\proctwo}{D[\procone,\proctwo]\langle\tail{m},\tail{n}\rangle}} \\
	E[\procone,\proctwo](m,n) &\defeq \ite{\eqe{m}}{\ite{\eqe{n}}{\procone}{\proctwo}}{\ite{\eqe{n}}{\proctwo}{E[\procone,\proctwo]\langle\tail{m},\tail{n}\rangle}}	
\end{align*}
\sloppy{These definitions are parametric in two arbitrary processes $\procone,\proctwo$. Given two integers $m,n$ represented as lists of zeros, the process $D[\procone,\proctwo]\langle m,n\rangle$ (resp.\ $E[\procone,\proctwo]\langle m,n\rangle$) evaluates to $\procone$ if $m\leq n$ (resp.\ $m=n$) and to $\proctwo$ otherwise.} Then, the memory process $M$ may be defined as follows:
\begin{align*}
	M'(m,n) &\defeq \para{C\langle m,\emstr\rangle}{E[\proczero,M'\langle\apz{m},n\rangle}]\langle m,n\rangle \\
	M(c) &\defeq \inpc{\mathtt{1}}{\evarone}D[M\langle c\rangle,\para{M'\langle\apz{c},x\rangle}{M\langle x\rangle}]\langle x,c\rangle.
\end{align*}
In other words, $M\langle c\rangle$ waits on channel $\mathtt 1$ for an integer $x$, which corresponds to the address of a memory cell. If $x\leq c$, the process returns to its initial state $M\langle c\rangle$. If $x>c$, the process goes to state $M\langle x\rangle$ and, in parallel, creates $x-c$ memory cells, each initialized to zero, at the addresses going from $c+1$ to $x$.

A RAM program is a finite sequence of instructions, which may be represented by mutually recursively 
defined process identifiers $I_1,\ldots,I_n$, $I_j$ standing for the $j$th instruction. These processes 
access the memory $M$ with the instructions allowed by the RAM (load/store operations, possibly with 
indirection), and do simple arithmetic operations (increment/decrement) on the contents of a special 
memory cell located at the channel $\emstr$, and called the \emph{accumulator}. Of course, before accessing the memory cell at address $c$, each instruction must take care of sending $c$ on channel $\mathtt 1$, which has the effect of creating the cell if it does not exist (and has no effect otherwise). The \textsc{halt} 
instruction corresponds to the process
$$
\para{\out\emstr{0\strone}}{\inpc\strone v\out\ochan v},
$$
which reads the value stored in the accumulator and forwards it to the external world through the 
output channel $\ochan$ (the string $\strone$ is arbitrary, as long as it is of length at least $2$ and starts with $1$ to avoid unwanted interferences).

Then, such a RAM program may be represented by the process
$$\label{equ:rampro}
\para{\inpc\ichan\evarone(\para{I_1}{C\langle\emstr,x\rangle})}{M\langle 0\rangle}.$$


Note that this encoding is not quite economic in terms of parallelism: a RAM is a sequential machine, whereas executing its encoding given above on the process machine will use several processors. However, it has the advantage of being easily generalized to PRAMs (see \refsect{PRAM}).

\subsection{Boolean Circuits}
\label{sec:Boolean}
The definitions of $\mathit{And}$ and $\mathit{Or}$ given above may be easily adapted to encode gates, from which a Boolean circuit 
is implemented immediately. For what concerns the interface, a circuit with $m$ inputs and $n$ outputs will be represented by a 
process reading bits from the external input channels $\ichan_1,\ldots,\ichan_m$ and sending bits to the external output channels 
$\ochan_1,\ldots,\ochan_n$.

Note that the behavior of a process representing a Boolean circuit as above is not functional (\refdef{FunBeh}). Therefore, although this encoding shows how circuits may be simulated on the process machine, it does not help extending \refth{FunClasses} to sublinear classes such as \classNC.

\subsection{Parallel Random Access Machines}
\label{sec:PRAM}
A PRAM is composed of several RAM programs running in parallel, each with its own accumulator. They access 
the same memory, including the accumulators of all other programs. At each step, the current instruction 
of every program is executed, and the machine proceeds to the next step only when the execution of all 
instructions is complete; in other words, the parallel components share a \emph{clock}. Concurrent access 
to memory is resolved on a first-come-first-served basis; it is the programmer's responsibility to ensure 
that the cooperation between the PRAM programs is consistent.

A PRAM program with $n$ parallel components is implemented by a process of the following form:
$$\para{\inpc{\ichan_1}{\evarone_1}(\para{J^1_1}{C\langle 0,x_1\rangle})\st\cdots\st
\inpc{\ichan_n}{\evarone_n}(\para{J^n_1}{C\langle 0^n,x_n\rangle})}{\para{M\langle 0^{n+1}\rangle}{K_n}}.$$
The idea is that we put the encodings of the $n$ RAM programs in parallel, plus a clock process $K_n$. 
The $i$th program has an associated internal channel, say $1^i$, with which it communicates with the clock.

If $I_j$ is the process encoding the $j$th instruction of a RAM program, the same instruction in the $i$th 
component of the PRAM is encoded by a process $J^i_j$ of the form $$1^i.I',$$
where the input prefix $1^i$ bounds a variable which does not appear in $I'$, and where $I'$ is $I$ in which 
a bogus string is sent on channel $1^i$ upon completion of the instruction. The encoding of the \textsc{halt} 
instruction is defined so that the $i$th program sends the contents of its accumulator through the external 
output channel $\ochan_i$.

The clock may then be represented by the process defined as follows:
$$K_n\defeq\overline{1}.\ldots.\overline{1^n}.1.\ldots.1^n.K_n,$$
where an output action of the form $\overline{1^i}$ means that the data sent is irrelevant.

In contrast with the representation of Boolean circuits (\refsect{Boolean}), the above encoding of PRAMs does yield functional behaviors. However, since the initial input instruction has a linear cost in the length of the input string, no process representing a PRAM runs in sublinear time. With the present definition of functional behavior, sublinear time classes (such as \classNC) may be captured only if we modify the process machine, for example allowing random access to the bits of the input string.

\section{Proof of \refth{FunClasses}}
\label{sec:MainThm}
Points (1) and (3) are consequences of the encoding of \ref{sec:Turing} and \ref{sec:AltTuring}.

For what concerns point (2), suppose there is a process $\procone$ deciding a language in time $f(n)$ and space 
$g(n)$ on the process machine. First of all, observe that, since $\labsem\procone$ is functional, the 
non-determinism that may be present in $\procone$ is vacuous. Indeed, all choices made during the 
execution of $\procone$ with a given input yield the same output; moreover, since $\bisim_d$ does not 
introduce divergence, the execution of $\procone$ terminates no matter what choice is made. Therefore, 
a deterministic Turing machine may simulate the process machine executing $\procone$ by simulating the 
transitions in any order (for instance, since the coding of a configuration $\cfg{\msetone}{\qfunone}$ 
will actually represent the set $\msetone$ as a list, we may choose to always execute the transition 
given by the first processor of the list).

We proceed to define the Turing machine simulating the execution of $P$. We start by fixing an encoding 
$\encode{(\cdot)}$ of configurations of the process machine as strings of a suitable (finite) alphabet. 
This may safely be supposed to satisfy, for every configuration $\cfgone$, $\len{\encode{\cfgone}}=k\len{\cfgone}$, 
where $k$ is a positive constant. Moreover, we suppose that $\encode{\cfgone}$ has a distinguished processor 
among the ones active in $\cfgone$. The Turing machine is initialized with a binary string $\strone$ on its 
input tape, and the string $\encode{(\cfg{(P,\emptyset)_\emstr}{\epsilon})}$ on its work tape, with the distinguished 
processor being the only active one.

Now, at each step, the Turing machine looks at the state $(\proctwo,\memone)_\ptag$ of the distinguished 
processor. Depending on the shape of $\proctwo$, the Turing machine simulates the appropriate transition. 
If $\proctwo=\inpc{\ichan}{\evarone}\procthree$ (by bisimulation, the channel must be $\ichan$), the Turing 
machine assigns to $\evarone$ the string placed on its input tape. If $\proctwo=\outc{\ochan}{\expone}\procthree$ 
(by bisimulation, the channel must be $\ochan$) the Turing machine halts; it accepts iff $\val\memone\expone=\emstr$ 
(remember the convention used in \refdef{Lang}). In all other cases, the Turing machine simulates the necessary 
operations and updates the encoding of the configuration accordingly. There are only two cases worth of attention. 
The first one is that in which $\proctwo=\inpc{\expone}{\evarone}\procthree$ and the queue $\val\memone\expone$ 
is empty. Then, the Turing machine simply selects a new distinguished processor (taking care of finding
one which is not blocked), and simulate the next step from it. The second one is that in which 
$\proctwo=\para{\procthree}{\procfour}$. In that case, after the spawning is simulated, the distinguished 
processor is chosen to be the one executing $\procthree$.

Thanks to the properties of $\bisim_d$ (in particular the fact that it does not introduce divergence), the 
above Turing machine is guaranteed to terminate with the correct state of acceptance/rejection w.r.t.\ the 
input string, which we suppose to be of length $n$. Note that, globally, the Turing machine actually simulates 
a run $\cfgone_0\mtautrans\cfgone_1\mtautrans\cdots\mtautrans\cfgone_t$. Consider now the tree defined as 
follows: the nodes at level $0\leq i\leq t$ represent the active processors in $\cfgone_i$; the root 
represents the only active processor in $\cfgone_0$ and, at each level, the siblings of a node are either 
one node (the processor did not spawn) or two nodes (the processor spawned), or none at all (the processor 
went idle). Now, obviously each transition in the run requires an active processor; therefore, $t$ is 
bounded by the size of the above tree. But such a tree has height bounded by $f(n)$ and width bounded 
by $g(n)$, so $t\leq f(n)g(n)$. The simulation of a single transition $\cfgone_i\mtautrans\cfgone_{i+1}$
of the run may be assumed to require at most something like $c\cdot{\max(\len{\cfgone_i},\len{\cfgone_{i+1}})}^h$ 
Turing machine steps, where $c,h$ are positive constants, which is bounded by $ck\cdot g(n)^h$. 
Hence, the runtime of the Turing machine is bounded by $ck\cdot f(n)g(n)^{h+1}$.

\section{Proof of \refth{Serv}}
\label{sec:ServThm}
Let $h:\Words\rightarrow\Words$. Point 1 is obvious: from a process $Q$ implementing $\behserv_h$ in time $f$ and space $g$, we syntactically extract a process implementing $\behfun_h$ by simply tracing the execution of $Q$ after the first input is given, halting immediately after the output is computed. The time and space bounds are obviously the same.

For point 2, we take an implementation $P$ of the functional behavior $\behfun_h$ and transform it into a server process which reads an external input, spawns a copy of $P$ to compute $h$ and waits for a further input. If $o_n$ is the event corresponding to the output of $h(\strone_n)$ after the strings $\strone_1,\ldots,\strone_n$ have been read, we have that any chain of events whose maximum is $o_n$ has the form $i_1<e^\tau_1<\cdots<i_n<e^\tau_n<d^\tau_1<\cdots<d^\tau_m<o_n$, where the $i_j$ are the input events reading $\strone_j$, the $e^\tau_j$ are spawn instructions, and $d^\tau_1,\ldots,d^\tau_m$ are events occurring in the computation of $h(\strone_n)$, so their total time cost is at most $f(\len{\strone_n})$. Therefore, assuming $f$ to be monotonic and at least linear, we have $\tcost{o_n}=\bigo(f(\len{\strone_n})+\sum_{j=1}^n(1+\len{\strone_j}))=\bigo(f(\inpsize{o_n})+\inpsize{o_n})=\bigo(f(\inpsize{o_n}))$. On the other hand, the space cost is identical.

\section{Reactive Turing machines}
\label{sec:RTM}
\newcommand{\lts}{\mbox{\ensuremath{\text{\textsc{lts}}}}}
\let\oldTo\to
\renewcommand{\to}[1][]{%
  \xrightarrow{#1}%
}
\renewcommand{\ite}[3]{%
  [#1].\left(#2,#3\right)%
}
\newcommand{\Sum}[4]{%
  \SumExpl{1 \leq #1 \leq #2}{#3}{{#4}_i}
}
\newcommand{\SumExpl}[3]{%
  \sum_{#1}^{#2}{#3}%
}

\newcommand{\Actions}{\ensuremath{\mathcal{A}}}
\newcommand{\actone}{a}
\newcommand{\acttwo}{b}
\newcommand{\bcode}[1]{%
  \underline{#1}%
}

We base our discussion of reactive Turing machines~\cite{BaetenLuttikTilburg}
on the following two definitions from~\cite{DBLP:conf/icdcit/BaetenLT11}
because the latter work mentions explicitly the alphabets
(while the former work leaves this information implicit). 

\begin{definition}[Reactive Turing machine]
  \label{def:RTM}
  A \emph{Reactive Turing machine}
  is a six-tuple 
  $M = (S,\mathcal{A},\mathcal{D},\to, \uparrow, \downarrow)$ where:
  \begin{enumerate}
  \item
    $S$ is a finite set of states,
  \item 
    $\mathcal{A}$ is a finite action alphabet, 
    $\mathcal{A}_{\tau}$ also includes the silent step ${\tau}$,
  \item 
    $\mathcal{D}$ is a finite data alphabet, 
    we add a special symbol ${\Box}$ standing for a blank
    and put $\mathcal{D}_{\Box} = \mathcal{D} \cup  \{{\Box}\}$,
  \item
    ${\to} \subseteq {
      S \times
      \mathcal{A}_{\tau} \times
      \mathcal{D}_{\Box}
      \times
      \mathcal{D}_{\Box} 
      \times
      \{L,R\}
      \times S}$ is a finite set of transitions or steps,
  \item 
    ${\uparrow} \in  S$ is the initial state,
  \item 
    ${\downarrow} \subseteq  S$ is the set of final states.
  \end{enumerate}
\end{definition}

Intuitively, 
the machine starts at the initial state $\uparrow$ with an empty tape; 
as usual a configuration consists of the state and the tape contents
(an element of $\mathcal{D}_{\Box}^*$) 
and a position of the read-write head with a single symbol. 
The possible transitions of a configuration depend on the state and the symbol under the head. 
Each transition comes with a (possibly observable) action
and a change on the tape as usual. 

\begin{definition}[LTS of an RTM]
  \label{def:ltsRTM}
  Let
  $M = (S,\mathcal{A},\mathcal{D},\to, \uparrow, \downarrow)$
  be an \textsc{rtm}. 
  The labeled transition system of $M$, denoted by $\mathcal{T} (M)$, 
  is defined as follows. 
  \begin{enumerate}
  \item     
    The set of states is the set of configurations
    $\{(s, \delta ) \mid  s \in  S, \delta  \text{ a tape instance} \}$.
  \item 
    The transition relation $\to$ is the least relation that satisfies 
    the following two properties
    for all $a \in  A_{\tau}$, $d, e \in  D_{\Box}$, and $\delta , \zeta \in D^{*}_{\Box}$.
    \begin{itemize}
    \item 
      $(s, \delta \bar d \zeta) \to[a] (t, (\delta\bar{\,})e\zeta)$ 
      iff 
      $s \to[{a[d/e]L}]t$,
    \item 
      $(s, \delta \bar d \zeta) \to[a] (t, \delta e (\bar{\,}\zeta))$
      iff $s\to[{a[d/e]R}] t$.
    \end{itemize}
  \item 
    The initial state is $(\uparrow , \bar{\Box})$.
  \item 
    $(s, \delta ) \downarrow$ iff $s \downarrow$.
  \end{enumerate}
\end{definition}
We have not introduced all notation as we have the usual use of a ``marker'' on tape symbols, 
i.e., the marked symbol $\bar d$ for each $d \in \mathcal{D}$.

For the encoding of \textsc{rtm}s into processes, 
we assume that the action alphabet $\Actions$ has a suitable coding into words, i.e., 
that each $\actone \in \mathcal{A}$ has a binary code $\bcode{\actone} \in \Words \setminus \{ \emstr\}$
such that $\bcode{\actone}$ is a palindrome\footnote{%
  This makes the encoding simpler and causes only a constant factor in size and time.%
}. 
Now, we shall code each action
$\actone\in\Actions$ 
by the prefix
$\out{\ochan}{\bcode{\actone}}$. 
First, we ignore termination, 
which is a minor point 
as discussed in Remark~\ref{rem:termination}. 
Similarly, states and data symbols are assumed to have binary codes. 

The encoding 
will use of the following auxiliary definition of internal choice.
\begin{definition}[Choice at a channel]
  \label{def:choiceAtChannel}
  Let $\strone\in\Words$ be a word, 
  let $n\geq 0$ be a natural number, 
  and let $\procone_1, \dots, \procone_n$ be a finite family of processes. 
  Now, we define the process $\Sum{i}{n}{\strone}{\procone}$ inductively as follows.
  \begin{align*}
    \Sum{i}{0}{\strone}{\procone}
    &= \proczero 
    \\
    \Sum{i}{m+1}{\strone}{\procone}
    &= 
    \para{%
      \para{%
        \outc{\strone}{\mathtt{0}}\proczero%
      }{%
        \outc{\strone}{\mathtt{1}}\proczero%
      }
    }{%
      \inpc{\strone}{\evarone}\inpc{\strone}{\evartwo}\ite{\eqz{\evartwo}}{\Sum{i}{m}{\strone}{\procone}}{\procone_{m+1}}
    }
  \end{align*}
\end{definition}
If $\strone$ is omitted, it is automatically $\emstr$, 
i.e., $\Sum{i}{m}{}{\procone}$ stands for $\Sum{i}{m}{\emstr}{\procone}$. 
We write $\procone_1 +^{\strone} \procone_2$ for $\Sum{i}{2}{\strone}{\procone}$
and $\procone_1 + \procone_2$ for $\Sum{i}{2}{}{\procone}$.

\newcommand{\head}{%
  d%
}
\newcommand{\daeh}{%
  e%
}

\newcommand{\lstr}{%
  \phi%
}
\newcommand{\rstr}{%
  \psi%
}
\newcommand{\xstr}{%
  \strone%
}
\newcommand{\lcel}{%
  \sigma%
}
\newcommand{\rcel}{%
  \rho%
}
\newcommand{\xcel}{%
  \theta%
}

\newcommand{\rev}[1]{%
  \overleftarrow{#1}%
}

The encoding of an \textsc{rtm}  will be a recursive process with six string parameters:
besides the state $\stateone$, 
the symbol on the tape under the head $\head$, 
and the left and right tape contents $\lstr$ and $\rstr$, 
we will use two auxiliary strings $\lcel, \rcel \in \Words$, 
which ``divide'' $\lstr$ and $\rstr$ into cells of suitable length. 
We shall also have an auxiliary processes for the manipulation of pairs of these strings, 
each of which is used to represent a stack: 
the first string of a pair contains the actual stack content
and the other string is the $\mathtt{0}$-separated list of the 
unary encodings of the lengths of the elements in the stack.
For example the pair $\langle \mathtt{01100000}, \mathtt{1111011110}\rangle$
is a stack with element $\mathtt{0110}$ on top of element $\mathtt{0000}$. 
The complete encoding of \textsc{rtm}s is now as follows.
\begin{definition}[Encoding RTMs without termination]
  \label{def:encodingRTM}
  Given a reactive Turing machine $M = (S,\mathcal{A},\mathcal{D},\to, \uparrow, \downarrow)$
  with ${\downarrow} = \emptyset$, 
  its encoding is the process $T\langle\bcode{\stateone}, \emstr, \emstr,\bcode{\Box}, \emstr,\emstr\rangle$, 
  with the following definitions. 
  \begin{align*}
    &\!\!\!\!\!\!\!\!\!T(\stateone, \lstr, \lcel, \head, \rstr,\rcel) =  \\
    &\!\!\!\!\!\!\!\!\!\SumExpl{s\to[{\actone[\head/\daeh]R}]t}{}{%
      \outc{\ochan}{\bcode{\actone}}
      (%
      \para{\mathit{Pop}\langle\rstr,\rcel,\emstr\rangle}{
        \inpc{\mathtt{0}}{\head'}%
        \inpc{\mathtt{0}}{\rstr'}%
        \inpc{\mathtt{0}}{\rcel'}%
        T\langle\bcode{t}, \bcode{\daeh}\lstr, \mathtt{1}^{\len{\bcode{\daeh}}}\mathtt{0}\lcel, \head', \rstr', \rcel' \rangle
      }
      )
    }
    +{} \\
    &\!\!\!\!\!\!\!\!\!\SumExpl{s\to[{\actone[\head/\daeh]L}]t}{}{%
      \outc{\ochan}{\bcode{\actone}}
      (%
      \para{\mathit{Pop}\langle\lstr,\lcel,\emstr\rangle}{
        \inpc{\mathtt{0}}{\head'}%
        \inpc{\mathtt{0}}{\lstr'}%
        \inpc{\mathtt{0}}{\lcel'}%
        T\langle\bcode{t},  \lstr', \lcel', \head', \bcode{\daeh}\rstr, \mathtt{1}^{\len{\bcode{\daeh}}}\mathtt{0}\rcel \rangle
      }
      )
    }
    \\
    \mathit{Pop}(\xstr, \xcel, \strtwo) 
    &= 
    \ite
    {\eqe{\xcel}}
    {%
      \outc{\mathtt{0}}{\bcode{\Box}}
      \outc{\mathtt{0}}{\emstr}
      \outc{\mathtt{0}}{\emstr}
      \proczero
    }
    {
      \ite
      {\eqz{\xcel}}
      {
        \outc{\mathtt{0}}{\strtwo}
        \outc{\mathtt{0}}{\xstr}
        \outc{\mathtt{0}}{\tail{\xcel}}
        \proczero
      }
      {
        \procone
      }
    }%
    \intertext{ where $\procone$  is
    }
    \procone
    &= 
    \ite
    {\eqz{\xcel}}
    {%
      \mathit{Pop}\langle\tail{\xstr},\tail{\xcel}, \mathtt{0}\strtwo\rangle
    }
    {%
      \mathit{Pop}\langle\tail{\xstr},\tail{\xcel}, \mathtt{1}\strtwo\rangle
    }. 
  \end{align*}
\end{definition}%
The $\mathit{Pop}$-process takes a pair of words -- representing a stack -- and an ``accumulator'' 
for the bit-wise compilation of the top element; 
it ``returns'' the top symbol and the updated ``stack'' via consecutive outputs on channel $\mathtt{0}$. 


 \begin{proposition}[Simulation of RTMs]
   \label{prop:sim}
   Every \textsc{rtm} with the empty set of final states 
   can be simulated by a process. 
 \end{proposition}
 \begin{proof}[idea]
   All non-determinism of the \textsc{rtm} is captured 
   at the root of the definition of the process identifier $T$
   of Definition~\ref{def:encodingRTM}. 
   Now, 
   it remains only to verify that the rest of the process actually encodes the reactive Turing machine as expected. 
   Note that there will never be any ``garbage'' sending actions  on the ``channels'' $\mathtt{0}$ and $\emstr$. 
 \end{proof}

 \begin{remark}[Termination]
   \label{rem:termination}
   As for termination, 
   without changing much, 
   we could add a new process constant $\mathbf{1}$ to our process syntax; 
   it would essentially behave as $\mathbf{0}$ 
   with the sole difference that if all processes in a configuration have reduced to $\mathbf{1}$
   (in an arbitrary environment and with arbitrary queues), 
   then the configuration is a final configuration. 
 \end{remark}





\paragraph{Effective  transition systems}
An \textsc{lts} is effective,
according to \cite[Definition~5]{BaetenLuttikTilburg},
if the transition relation is recursively enumerable.
Note that we again ignore termination for the sake of simplicity (cf.~Remark~\ref{rem:termination}). 
Clearly, the \textsc{lts}s of processes are recursively enumerable.

\end{document}